\documentclass[journal,10pt]{IEEEtran}
\usepackage[OT1]{fontenc}
\usepackage{cite}
\usepackage{dsfont}
\usepackage{amsfonts}
\usepackage{amsmath}
\usepackage{amssymb}
\usepackage{amsthm}
\newtheorem{thm}{Theorem}
\newtheorem{col}{Corollary}

\usepackage{enumitem}
\usepackage{graphicx}
\usepackage{subfigure}
\usepackage{xcolor}
\usepackage{caption}
\usepackage[bookmarks=false, draft]{hyperref}
\usepackage{stfloats}

\begin{document}
\title{On the Coverage and Capacity of Ultra-Dense Networks with Directional Transmissions}
	\author{Yining Xu,
          Sheng~Zhou,~\IEEEmembership{Member,~IEEE}
    \thanks{This work is sponsored in part by the Nature Science Foundation of China (No. 61871254, No. 91638204, No. 61861136003, No. 61571265, No. 61621091), and Intel. \emph{(Corresponding author: Sheng Zhou.)}}
    \thanks{The authors are with Beijing National Research Center for Information Science and Technology, Department of Electronic Engineering, Tsinghua University, Beijing 100084, China. Email: xu-yn16@mails.tsinghua.edu.cn, sheng.zhou@tsinghua.edu.cn.}
        }
\maketitle

\begin{abstract}
 We investigate the performance of a downlink ultra-dense network (UDN) with directional transmissions via stochastic geometry. Considering the dual-slope path loss model and sectored beamforming pattern, we derive the expressions and asymptotic characteristics of the coverage probability and constrained area spectrum efficiency (ASE). Several special scenarios, namely the physically feasible path loss model and adjustable beam pattern, are also analyzed. Although \emph{signal-to-interference-plus-noise ratio collapsing} still exists when the path loss exponent in the near-field is no larger than 2, using strategies like beam pattern adaption, can avoid the decrease of the coverage probability and constrained ASE even when the base station density approaches infinity.
\end{abstract}

\IEEEpeerreviewmaketitle
\section{Introduction}
Network densification, millimeter-wave (mmWave) communications and massive multiple-input multiple-output are regarded as the `big three' promising key technologies in 5G mobile communication systems \cite{6824752}. Performing directional transmissions in the mmWave band can help compensate for the strong path loss and reduce the inter-cell interference when moving towards ultra-dense networks (UDNs). Dual-slope path loss model shows good precision in mmWave UDNs \cite{7061455}, but it also brings difficulties in network performance analyses.

In this paper, we study the coverage and capacity performance of a downlink UDN with directional transmissions using stochastic geometry. For the coverage performance in the homogeneous cellular network with single-slope path loss model, ref.\cite{6042301} indicates that \emph{signal-to-interference-plus-noise ratio (SINR) invariance} exists. This property reveals that increasing the base station (BS) density does not affect the coverage probability under extremely dense deployment. The results with multi-slope path loss model are derived in \cite{7061455}\cite{7588290}. They find the near-field path loss exponent has a \emph{phase transition} feature, meaning that the throughput has different asymptotic characteristics depending on whether the near-field path loss exponent exceeds a threshold. A multi-slope path loss model with line-of-sight (LOS) and non-line-of-sight (NLOS) channels is investigated in \cite{8077766}. The results indicate that under sufficiently large network density, the coverage probability decreases and the constrained area spectrum efficiency (ASE) experiences a slow growth or \emph{even a decrease}. Ref.\cite{8375976} mainly focuses on the asymptotic characteristic of the ASE. Three possible definitions of the ASE and different path loss functions are studied. However, the expression of ASE is not derived and beamforming is not considered in this work. For the beamforming in the mmWave cellular networks, ref.\cite{6932503} assumes a single-slope path loss model with LOS and NLOS channels. While ref.\cite{8086182} considers the mmWave relaying under a two-ray path loss model, consisting of a LOS path and a reflection path. A Monte Carlo simulation of UDNs with directional transmissions is given in \cite{7962711} but the theoretical analysis is yet to be resolved. Our work differs from the aforementioned works in the following ways: 1) We derive the general expressions of the coverage probability and constrained ASE under the dual-slope path loss model in a downlink UDN with \emph{directional transmissions}. Additionally, the coverage probability is revealed to approach zero with dense BSs when the near-field path loss exponent is no larger than 2 and Rayleigh faded interference is considered. 2) We analyze the asymptotic characteristics under the physically feasible path-loss model and adjustable beam pattern. The monotonicity of coverage probability w.r.t the near-field intensity and constrained ASE w.r.t the beam alignment probability are proved for the former scenario. The beam pattern adaption scheme is proposed for the latter scenario.

\begin{figure}
  \centering
  \includegraphics[width=3.3in]{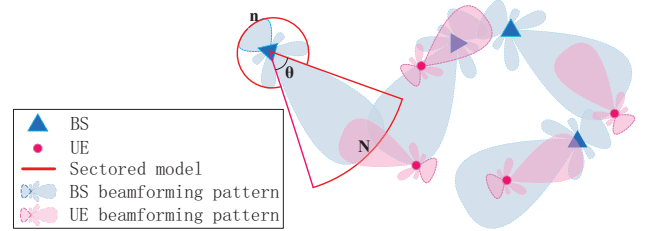}
  \setlength{\abovecaptionskip}{3pt}
  \setlength{\belowcaptionskip}{3pt}
  \caption{Sectored approximation of the beamforming pattern.}
  \vspace{-0.25in}
  \label{fig_net}
\end{figure}
\section{System Model}\label{secmodel}
Consider a downlink UDN where the locations of BSs form a Poisson point process (PPP) $\Phi$ with density $\lambda$ on a 2-dimensional infinite plane. The locations of user equipments (UEs) also form a PPP with density $\kappa$ which is independent of $\Phi$. Each UE associates with its nearest BS. The UEs associated to the same BS are allocated with orthogonal time-frequency resource blocks (RBs). The analysis is on a specific RB that is reused by all BSs. All BSs work in the full-loaded mode with the same transmission power $P_t$ and BS sleeping is not considered in this paper. The dual-slope path loss model we adopt is\cite{1232163}:
\newcounter{TempEqCnt}
\setcounter{TempEqCnt}{\value{equation}}
\setcounter{equation}{0}
\begin{figure*}[htb]
\hrule height 1pt
\vspace*{2pt}
\begin{thm}\label{thm1}
  The coverage probability of the UDN with directional transmissions under the dual-slope path loss model is
 \setlength{\abovedisplayskip}{2pt}
 \setlength{\belowdisplayskip}{2pt}
 \begin{align*}
   \begin{split}
    P_{c}(\lambda,T)
    \!=\!\!\!\int_0^{d_0}\!\!\!\!\emph{exp}\!\left(\!\!\frac{-\mu T\sigma^2}{N_B N_U\alpha_0r_0^{-\beta_1}}\!\!\right)\mathcal{L}_{I_{1}}\!\!\left(\!\!\frac{\mu T}{N_B N_U\alpha_0r_0^{-\beta_1}}\!\!\right) \!2\pi\lambda r_0\emph{d}r_0
    \!+\!\!\!\int_{d_0}^\infty \!\!\!\!\emph{exp}\!\left(\!\!\frac{-\mu T\sigma^2 d_0^{\beta_1-\beta_2}}{N_B N_U\alpha_0r_0^{-\beta_2}}\!\!\right)\mathcal{L}_{I_{2}}\!\!\left(\!\!\frac{\mu T d_0^{\beta_1-\beta_2}}{N_B N_U\alpha_0r_0^{-\beta_2}}\!\!\right) \!2\pi\lambda r_0\emph{d}r_0
   \end{split}
 \end{align*}
  where the functions $\mathcal{L}_{I_{1}}(x)$ and $\mathcal{L}_{I_{2}}(x)$ are given by
 \begin{align}\label{L_I_1}
  \begin{split}
    \mathcal{L}_{I_{1}}(x)
    =\prod_{k=1}^{4}\emph{exp}\Bigg(-&\frac{2\pi\lambda  b_k}{\beta_1}
    \left({x \alpha_0 a_k}\right)^{\frac{2}{\beta_1}}
    \int_{0}^\infty g_i^{\frac{2}{\beta_1}}
    \left(\Gamma\Big(-\frac{2}{\beta_1},x g_i \alpha_0 a_k r_0^{-\beta_1}\Big)
    -\Gamma\Big(-\frac{2}{\beta_1},x g_i \alpha_0 a_k d_0^{-\beta_1}\Big)\right)
    f(g_i)\emph{d}g_i\Bigg.\\[-5pt]
    &-\frac{2\pi\lambda b_k}{\beta_2}
    \left({x \alpha_0 a_k d_0^{\beta_2-\beta_1}}\right)^{\frac{2}{\beta_2}}
    \int_{0}^\infty g_i^{\frac{2}{\beta_2}}
    \left(\Gamma\Big(-\frac{2}{\beta_2},x g_i \alpha_0 a_k d_0^{-\beta_1}\Big)
    -\Gamma\Big(-\frac{2}{\beta_2}\Big)\right)
    f(g_i)\emph{d}g_i \Bigg)
  \end{split}
 \end{align}
 \setlength{\abovedisplayskip}{0pt}
 \begin{align}\label{L_I_2}
  \begin{split}
    \quad \!\mathcal{L}_{I_{2}}(x)
    = \prod_{k=1}^{4}\emph{exp}\Bigg(\!-\frac{2\pi\lambda b_k}{\beta_2}
    \left({x \alpha_0 a_k d_0^{\beta_2-\beta_1}}\right)^{\frac{2}{\beta_2}}
    \int_{0}^\infty g_i^{\frac{2}{\beta_2}}\left(\Gamma\Big(-\frac{2}{\beta_2},x g_i \alpha_0 a_k d_0^{\beta_2-\beta_1} r_0^{-\beta_2}\!\Big)-\Gamma\Big(-\frac{2}{\beta_2}\Big)\right)f(g_i)\emph{d}g_i\!\Bigg)
  \end{split}
\end{align}
  where $\Gamma(x,y) = \int_{y}^{\infty} t^{x-1}e^{-t} \emph{d}t$ is the incomplete gamma function and $\Gamma(x) = \Gamma(x,0)$ denotes the gamma function.
 \end{thm}
 \begin{proof}
 See Appendix A.
 \end{proof}
 \vspace*{-13pt}
 \hrulefill
 \vspace*{-10pt}
 \end{figure*}
 \setcounter{equation}{\value{TempEqCnt}}
 \setcounter{equation}{3}
  \setlength{\abovedisplayskip}{3pt}
  \setlength{\belowdisplayskip}{3pt}
\begin{align*}
  L(r)=
  \begin{cases}
    \alpha_0r^{-\beta_1} & \text{if  } 0<r<d_0\\
    \alpha_0r^{-\beta_2}/d_0^{\beta_1-\beta_2}& \text{if  } r\geq d_0,
  \end{cases}
\end{align*}
where $d_0$ is the Fresnel breakpoint, and $\alpha_0$ is the reference loss at 1m when $d_0 \geq 1\textrm{m}$. Let $\beta_1$ and $\beta_2$ denote the path loss exponents for the region closer than $d_0$ and away from it, respectively. We assume $\beta_2 > 2$ and $0 \leq \beta_1 \leq \beta_2$ for a finite average aggregated interference power. The beamforming pattern on either BS or UE is approximated to a sectored antenna model \cite{6932503}, shown in Fig. \ref{fig_net}. The main lobe gain, side lobe gain and beamwidth are denoted by $N_{j},n_{j}$ and $\theta_{j}$, where $j=\{\text{B,U}\}$ corresponds to the BS and UE, respectively. The front-back ratio of main lobe gain versus side lobe gain is denoted by $\epsilon_{j} = N_{j} / n_{j}$. Without loss of generality, a typical UE of interest is assumed to be located at the origin $o$, which associates with its nearest BS $B_o$ and their beam directions are aligned. An interfering BS $i$ is denoted by $B_i$ and the beamforming gain from $B_i$ to the typical UE is denoted by an independently and identically distributed (i.i.d) discrete random variable $G_i$ with the probability distribution
\begin{align*}
 \begin{split}
  &\mathbb{P}[G_i = N_B N_U] = \theta_B \theta_U/(4\pi^2)  ,  \\
  &\mathbb{P}[G_i = N_B n_U] = \theta_B (2\pi-\theta_U)/(4\pi^2)  ,\\
  &\mathbb{P}[G_i = n_B N_U] = (2\pi-\theta_B) \theta_U/(4\pi^2) , \\
  &\mathbb{P}[G_i = n_B n_U] = (2\pi-\theta_B) (2\pi-\theta_U)/(4\pi^2),
 \end{split}
\end{align*}
 and we adopt $\mathbb{P}[G_i = a_k] = b_k$, $k = 1,2,3,4$ to represent the equalities above for simplification. We use $h\sim \textrm{exp}(\mu)$ to model the small-scale Rayleigh fading of $B_o$ and an i.i.d random variable $g_i$ to denote the fading coefficient of $B_i$. Hence, the received SINR of the typical UE is:
\begin{align*}
  \text{SINR}&=\frac{N_B N_UL(r_0)h}{\sigma^2+I},
\end{align*}
where $r_0$ is the distance from $B_o$ to the typical UE, and $r_i$ is the distance from $B_i$, respectively. While $\sigma^2$ denotes the noise power normalized by $P_t$, and $I = \sum_{i:B_i\in\phi\backslash \{B_o\}} g_i G_i L(r_i)$ denotes the aggregated interference normalized by $P_t$.

\section{Main Results}\label{secresult}
\subsection{Coverage Probability}
Let $P_{c}(\!\lambda,T) \!=\! \mathbb{P}[\textrm{SINR}\!\!>\!T]$, denotes the coverage probability. Then we have Theorem \ref{thm1} shown at the top of this page.

We further simplify the result in Theorem \ref{thm1}. Firstly, we introduce the definition of \emph{physically feasible path loss} with following characteristics\cite{8375976}: 1) The path loss at zero distance is a finite number, i.e., $L(0)<\infty$; 2) The average received power at any point on the plane is no larger than the transmit power, i.e., $L(r)\leq L(0), \forall r\geq 0$; 3) The total received power from the points all over the plane is finite, i.e., $\int_0^{\infty} rL(r)\text{d}r <\infty$. For example, setting $\beta_1 = 0$ in the dual-slope path loss model satisfies the aforementioned characteristics.

\begin{col}\label{col1}
 When the following conditions are satisfied: 1) interference channel is Rayleigh faded and the noise is neglected, i.e., $g_i \sim \emph{exp}(\mu)$ and $\sigma^2 = 0$; 2) the side lobe gain of BSs and UEs are neglected, i.e., $n_{B}=n_{U}=0$; 3) physically feasible dual-slope path loss model, the coverage probability can be simplified to equality (\ref{simple_cov}), where $\rho(x,y)\!=\!x^{\frac{2}{y}}\!\!\int_{x^{-\frac{2}{y}}}^{\infty} \frac{1}{1+u^{\frac{y}{2}}} \emph{d} u$.
 \end{col}
\begin{proof}
 The main simplification is on the functions $\mathcal{L}_{I_{1}}(x)$ and $\mathcal{L}_{I_{2}}(x)$, given in (\ref{L_I_1}) and (\ref{L_I_2}). Utilizing the property of exponentially distributed $g_i$ and following the same step as (\emph{a}) of (\ref{P_c_2}) in the Appendix A, brings the Corollary \ref{col1}.
 \end{proof}
The coverage probability in Corollary \ref{col1} only depends on the following components: 1) The probability $\frac{\theta_B \theta_U}{4 \pi^2}$ that main lobe of BS points to main lobe of UE, defined as the \emph{beam alignment probability}. The coverage probability deceases monotonically with $\frac{\theta_B \theta_U}{4 \pi^2}$. This result indicates that wider beam with lower beamforming overhead leads to higher beam alignment probability but stronger interference. 2) The average number of BSs within the Fresnel breakpoint $\lambda \pi d_0^2$, defined as \emph{near-field intensity.} We prove the monotonicity of coverage probability w.r.t $\lambda \pi d_0^2$ in Appendix B. The coverage probability eventually drops to zero as $\lambda \pi d_0^2$ approaches infinity. 3) $\frac{T}{1+T}$ and $\rho(T,\beta_2)$, and they can be obtained when path loss characteristics and the minimal received SINR requirement are known.
\setcounter{TempEqCnt}{\value{equation}}
\setcounter{equation}{2}
\begin{figure*}[htb]
\vspace*{2pt}
\begin{align}\label{simple_cov}
  \begin{split}
     P_{c}(\lambda,T)
     =&\frac{1}{1-\frac{\theta_B \theta_U}{4 \pi^2}\frac{T}{1+T}} \text{exp}\Bigg(\!\!-\!\lambda \pi d_0^2\frac{\theta_B \theta_U}{4 \pi^2}\left(\rho(T,\beta_2)+\frac{T}{1+T}\right)\!\!\Bigg)\\
    &-\frac{\frac{\theta_B \theta_U}{4 \pi^2}\frac{T}{1+T}+\frac{\theta_B \theta_U}{4 \pi^2} \rho(T,\beta_2)}{\left(1-\frac{\theta_B \theta_U}{4 \pi^2}\frac{T}{1+T}\right)\Big(\frac{\theta_B \theta_U}{4 \pi^2} \rho(T,\beta_2)+1\Big)}
     \text{exp}\Bigg(\!\!-\!\lambda \pi d_0^2\left(\frac{\theta_B \theta_U}{4 \pi^2} \rho(T,\beta_2)+1\right)\!\!\Bigg)
  \end{split}
 \end{align}
\vspace*{2pt}
\begin{align}\label{general_ASE}
  \begin{split}
    A(\lambda,T)
    =&\frac{\lambda}{\text{ln2}}\int_0^{d_0}\!\!\int_{\text{ln}(1+T)}^\infty\!\!\text{exp}\!\left(\frac{-\mu (e^t-1)\sigma^2}{N_B N_U\alpha_0 r_0^{-\beta_1}}\right)
    \mathcal{L}_{I_{1}}\!\!\left(\frac{\mu (e^t-1)}{N_B N_U\alpha_0r_0^{-\beta_1}}\right) 2\pi\lambda r_0\text{d}t \text{d}r_0\\
    &+\frac{\lambda}{\text{ln2}}\int_{\!d_0}^\infty \!\!\int_{\text{ln}(1+T)}^\infty\!\!\text{exp} \!\left(\frac{-\mu (e^t-1)\sigma^2 d_0^{\beta_1-\beta_2}}{N_B N_U\alpha_0r_0^{-\beta_2}}\right)
    \mathcal{L}_{I_{2}}\!\!\left(\frac{\mu (e^t-1) d_0^{\beta_1-\beta_2}}
    {N_B N_U\alpha_0r_0^{-\beta_2}}\right) 2\pi\!\lambda r_0 \text{d}t \text{d}r_0
    +\lambda \text{log}_2 (1+T) P_c(\lambda,T)
  \end{split}
\end{align}
\vspace*{2pt}
\begin{align}\label{simple_ASE}
  \begin{split}
    A(\lambda,T)
    =\frac{\lambda}{\text{ln2}}\int_{\text{ln}(1+T)}^\infty & \left\{\frac{\text{exp}\Big(\!\!-\!\lambda \pi d_0^2 \frac{\theta_B \theta_U}{4\pi^2}\big(\frac{e^t-1}{e^t}+\rho(e^t-1,\beta_2)\big)\!\Big)
    -\text{exp}\Big(\!\!-\! \lambda \pi d_0^2 \big(1+\frac{\theta_B \theta_U}{4 \pi^2}\rho(e^t-1,\beta_2)\big)\!\Big)}
    {1-\frac{\theta_B \theta_U}{4 \pi^2}\frac{e^t-1}{e^t}}\right.\\
    &\left.\quad +\frac{\text{exp}\Big(\!\!-\! \lambda \pi d_0^2\big(1+\frac{\theta_B \theta_U}{4\pi^2}\rho(e^t-1,\beta_2)\big)\!\Big)}{\frac{\theta_B \theta_U}{4\pi^2} \rho(e^t-1,\beta_2)+1} \right\} \text{d} t
    +\lambda \text{log}_2 (1+T) P_c(\lambda,T)
  \end{split}
\end{align}
\vspace*{0pt}
\hrule height 1pt
\vspace*{-10pt}
 \end{figure*}
 \setcounter{equation}{\value{TempEqCnt}}
 \setcounter{equation}{5}
\subsection{Constrained Area Spectral Efficiency}
The constrained ASE is denoted by $A(\lambda, T) = \lambda \mathbb{E}[\text{log}_2(1+\text{SINR})\mathbb{I}\{\text{SINR}\geq T\}]$ \cite{8375976}, where $\mathbb{I}\{\cdot\}$ is an indication function, and we have following expressions of $A(\lambda, T)$.
\begin{thm}\label{thm2}
The constrained ASE of the UDN with directional transmissions under the dual-slope path loss model is given in equality (\ref{general_ASE}), where $\mathcal{L}_{I_{1}}(x)$ and $\mathcal{L}_{I_{2}}(x)$ are already given in Theorem \ref{thm1}.
\end{thm}
\begin{proof}
The derivation is similar to that of Theorem \ref{thm1}.
\end{proof}

\begin{col}\label{col2}
When the conditions in Corollary \ref{col1} are satisfied, the constrained ASE can be simplified to equality (\ref{simple_ASE}), where $\rho(x,y)$ is already defined in Corollary \ref{col1}.
\end{col}
\begin{proof}
The derivation is similar to that of Corollary \ref{col1}.
\end{proof}
 The constrained ASE also relies on the beam alignment probability, the near-field intensity and other terms related to $T$ and $\beta_2$. Corollary \ref{col2} indicates that the constrained ASE approaches zero as the $\lambda \pi d_0^2$ tends to infinity. Since $A(\lambda,T) \geq 0$  $\forall \lambda \in \mathbb{R}_{+}$, $A(0,T) = 0$ and $\lim_{\lambda\rightarrow\infty} A(\lambda,T) = 0$, the optimal BS density that maximizes the constrained ASE is finite. Meanwhile, the beam alignment probability $\frac{\theta_B \theta_U}{4 \pi^2}$ still results in monotonic decline of the constrained ASE and the proof is given in Appendix C.

\section{Asymptotic Characteristic Analyses}\label{secasy}
In this section, we study the asymptotic characteristics of the coverage probability and constrained ASE when the BS density approaches infinity. In addition, the idea of keeping constant SINR and linear growth of the constrained ASE by tuning the beam pattern w.r.t the BS density is elaborated.

\begin{thm}\label{thm3}
When the interference is Rayleigh faded, i.e., $g_i \sim \emph{exp}(\mu)$, the coverage probability of the UDN with directional transmissions under the dual-slope path loss model tends to zero as the BS density approaches infinity when $\beta_1 \leq 2$.
\end{thm}

\begin{proof}
The received SINR increases when 1) ignoring the noise, 2) reducing the side lobe gain to zero and 3) neglecting the interfering BSs located away from the Fresnel breakpoint. It can be expressed as $\text{SINR} \leq \text{SIR} \leq \text{SIR}_{n_B = n_U = 0; \; g_i G_i L(r_i) = 0, \forall r_i>d_0}$. Applying these conditions to the coverage probability expression in Theorem \ref{thm1}, leads to
\begin{align*}
 \begin{split}
    &\mathbb{P}\left[\text{SIR}>T |n_B = n_U = 0; \; g_i G_i L(r_i) = 0, \forall r_i>d_0\right]\\
    =&\lambda\pi \!\!\!\int_0^{d_0^2}\!\!\!\!\!\text{exp}\!\!\left(\!\!-\!\lambda \pi\frac{\theta_B \theta_U}{4 \pi^2} u \!\Bigg(\!\!1\!\!+ \!\!\!\int_{1}^{\frac{d_0^2}{u}}\!\!\!\!\!\!\frac{T}{T+(t)^{\frac{\beta_1}{2}}} \text{d}t\!\!\Bigg) \!\!\!\right) \!\! \text{d} u
    \!+\!\frac{\text{exp}(\!-\!\lambda \pi d_0\!)}{\lambda \pi d_0}.\\
 \end{split}
\end{align*}
Using the techniques proposed in Appendix B of \cite{7061455}, we get the result that the equality above tends to zero as $\lambda \rightarrow \infty$ and thus $\lim_{\lambda \rightarrow \infty} \mathbb{P}[\text{SINR} > T] = 0$ when $\beta_1 \leq 2$.
\end{proof}

\begin{thm}\label{thm4}
Considering the physically feasible dual-slope path loss model, i.e., $\beta_1 = 0$, and assuming Rayleigh faded interference, i.e., $g_i \sim \emph{exp}(\mu)$, the coverage probability and constrained ASE tend to zero as the BS density approaches infinity.
\end{thm}
\begin{proof}
As used in Appendix A of\cite{8375976}, let $\lambda = M\lambda_0$, where $M\in\mathbb{Z}_+$ and $\lambda_0\in\mathbb{R}_{+}$. Then the BSs deployment process $\phi$ can be divided into $M$ PPP process $\varphi$ with density $\lambda_0$. And the expression $\lim_{\lambda\rightarrow\infty}\lambda \text{SINR}(\lambda)$ is calculated as:
\begin{align*}
 \begin{split}
    \lim_{M\!\rightarrow\infty}\!\!\!\!M\!\lambda_0 \text{SINR}(\!M\!\lambda_0\!)\!\!
    \overset{(a)}=& \!\!\!\lim_{M\!\rightarrow\infty} \!\!\!\lambda_0 \frac{N_B N_U \alpha_0 h}{\frac{1}{M}\!\!\!\sum\limits_{m=1}^{M}\sum\limits_{r_{i,m}\in\varphi_m} \!\!\!\!g_{i,m} G_{i,m} L(r_{i,m})}\\
    \overset{(b)}=& \frac{N_B N_U\alpha_0 h}{ \mathbb{E}[G_i] \cdot2\pi\mu\int_{0}^{\infty}r L(r) \text{d}r} <\infty ,\\
 \end{split}
\end{align*}
where equality (\emph{a}) follows the fact that $L(r_0)\rightarrow\alpha_0$ and $(\sigma^2-N_B N_U L(r_0) h)/M\rightarrow0$ when $M\rightarrow\infty$. Step (\emph{b}) follows from the law of large numbers and the Campbell's theorem. Hence, $\lim_{\lambda\rightarrow\infty}P_c(\lambda,T) = 0$ is obtained. We can also have $\lim_{\lambda\rightarrow\infty}A(\lambda,T) = 0$ with similar steps, which are:
\begin{align*}
  \begin{split}
    &\lim_{M\rightarrow\infty} \!\!M \lambda_0 \mathbb{E}[\text{log}_2(\!1\!\!+\!\text{SINR}(\!M \lambda_0\!)\!)\mathbb{I}\{\text{SINR}(\!M \lambda_0\!)\!\!\geq\!\!T\}]\\
    \leq&\lim_{M\rightarrow\infty} \!\!\mathbb{E}\!\!\left[\frac{M \lambda_0}{\text{ln}(2)}\text{SINR}(M \lambda_0)\mathbb{I}\{\text{SINR}(M \lambda_0)\geq T\}\right]\!\!
    = 0.
  \end{split}
\end{align*}
\vspace*{-4pt}
\end{proof}

Based on Theorem \ref{thm4}, to keep constant coverage probability and linear growth of the constrained ASE when the BS density approaches infinity, the beam pattern shall be adapted as $\frac{\mathbb{E}[G_i]}{N_B N_U} \!=\!\frac{\theta_B\theta_U}{4\pi^2} + \frac{\theta_B (2\pi-\theta_U\!)}{4\pi^2\epsilon_U} + \frac{\theta_U (2\pi-\theta_B\!)}{4\pi^2\epsilon_B} + \frac{(2\pi-\theta_U\!) (2\pi-\theta_B\!)}{4\pi^2\epsilon_B \epsilon_U} \!=\! \frac{K}{\lambda}$, which indicates that the normalized expectation of beamforming gain is inversely proportional to the BS density. The constant $K$ can be decided based on the following limits of coverage probability and slope of constrained ASE w.r.t the BS density.

\begin{col}\label{col3}
If the conditions in Theorem \ref{thm4} are satisfied and beam pattern adaption $\frac{\mathbb{E}[G_i]}{N_B N_U}=\frac{K}{\lambda}$ is adopted, when the BS density approaches infinity, limits of coverage probability and slope of constrained ASE w.r.t the BS density are given as:
\begin{align}\label{lim_of_CP}
  \begin{split}
    &\lim_{\lambda \rightarrow \infty}P_c(\lambda,T) = \emph{exp}\!\left(-\frac{2 K \pi \mu^2 T \gamma}{\alpha_0}\right),
  \end{split}
\end{align}
\begin{align}\label{slope_of_ASE}
  \begin{split}
    &\lim_{\lambda \rightarrow \infty}A(\lambda,T) = \lambda \emph{ln}\!\left(\!\frac{\alpha_0 h}{2 \pi \mu \gamma}-1\!\right) \emph{exp}\!\left(\!-\frac{2 K\pi \mu^2 T \gamma}{\alpha_0}\!\right) ,
  \end{split}
\end{align}
where $\gamma = \int_0^{\infty} rL(r) \emph{d} r < \infty$.
\end{col}

\begin{proof}
The result is directly obtained from Theorem \ref{thm4}.
\end{proof}

The result indicates that unlike the cases without beam pattern adaptation that have diminishing coverage probability when the BS density grows, we can maintain a certain level of network performance by adjusting the beamwidth or the front-back ratio according to the BS density. For example, when the UEs have the same beamwidth and the same front-back ratio with the BSs, i.e., $\theta_B = \theta_U = \theta$ and $\epsilon_B = \epsilon_U = \epsilon$, the beam pattern adjustment principle is $\theta = 2 \pi (\epsilon \sqrt{K/ \lambda}-1)/(\epsilon - 1)$, where $K$ is calculated according to the desired limit of coverage probability or the slope of constrained ASE w.r.t the BS density in equations (\ref{lim_of_CP}) and (\ref{slope_of_ASE}), respectively.

\setcounter{TempEqCnt}{\value{equation}}
\setcounter{equation}{7}
\begin{figure*}[htb]
\vspace*{2pt}
\begin{align}\label{P_c_1}
  \begin{split}
   P_{c}(\lambda,T)
   =\!\int_0^{d_0}\!\!\!\mathbb{P}\!\!\left[h\!\!>\!\frac{T(\sigma^2+I)}{N_B N_U\alpha_0 r_0^{-\beta_1}}\!\Big|r_0\!\!\in\!\!(0,d_0)\!\right]\!\!e^{-\lambda \pi r_0^2} 2\pi\lambda r_0\text{d}r_0
   +\!\!\!\int_{d_0}^\infty\!\!\!\mathbb{P}\!\!\left[h\!\!>\!\!\frac{T(\!\sigma^2\!\!+\!\!I)d_0^{\beta_1\!-\!\beta_2}}{N_B N_U\alpha_0 r_0^{-\beta_2}}\!\Big|r_0\!\!\in\!\![d_0,\infty)\!\right]\!\!e^{-\lambda \pi r_0^2} 2\pi\lambda r_0\text{d}r_0
  \end{split}
\end{align}
\vspace*{2pt}
  \begin{align}\label{P_c_2}
    \begin{split}
      \int_{r_0}^\infty\!\!\left(\!\!1\!-\!\!\sum_{k=1}^4 \!b_k\text{exp}\!\left(\!-xg_ia_k L(v)\!\right)\!\!\right)\!\!v\text{d}v\!
      \overset{(a)}{=}\!\!\int_{r_0}^{d_0}\!\!\left(\!\!1\!-\!\!\sum_{k=1}^4 \!b_k\text{exp}\!\left(\!-xg_ia_k\alpha_0v^{-\beta_1}\!\right)\!\!\!\right)\!\!v\text{d}v
      +\!\!\int_{d_0}^\infty\!\!\left(\!\!1\!-\!\!\sum_{k=1}^4 \!b_k\text{exp}\!\left(\!-xg_ia_k\alpha_0v^{-\beta_2}d_0^{\beta_2\!-\!\beta_1}\!\right)\!\!\!\right)\!\!v\text{d}v\\
    \end{split}
  \end{align}
\vspace*{0pt}
\hrule height 1pt
 \vspace*{-10pt}
 \end{figure*}
 \setcounter{equation}{\value{TempEqCnt}}

\section{Numerical and Simulation Results}\label{secnum}
The default system parameters are: $\beta_1 = 2$, $\beta_2 = 4$, $d_0 = 10 \text{m}$, $\alpha_0 = 0\text{dB}$, $\mu = 1$, $T = 7\text{dB}$ and $\sigma^2 =\alpha_0 d_0^{-\beta_1}/100$ as the signal-to-noise ratio at $d_0$ is set as $20\text{dB}$. For the beamforming patterns, we assume that $N_B = 20\text{dB}, n_B = 0\text{dB}, \theta_B = \frac{\pi}{6}, N_U = 10\text{dB}, n_U = -10\text{dB}$ and $\theta_U = \frac{\pi}{2}$.

\begin{figure}[!ht]
  \setlength{\abovecaptionskip}{-2pt}
  \setlength{\belowcaptionskip}{-6pt}
  \centering
  \vspace{-0.15in}
  \subfigure[$P_{c}({\lambda},T)$ scaling with $\lambda$]{\label{Cov_general}
    \includegraphics[width=0.231\textwidth]{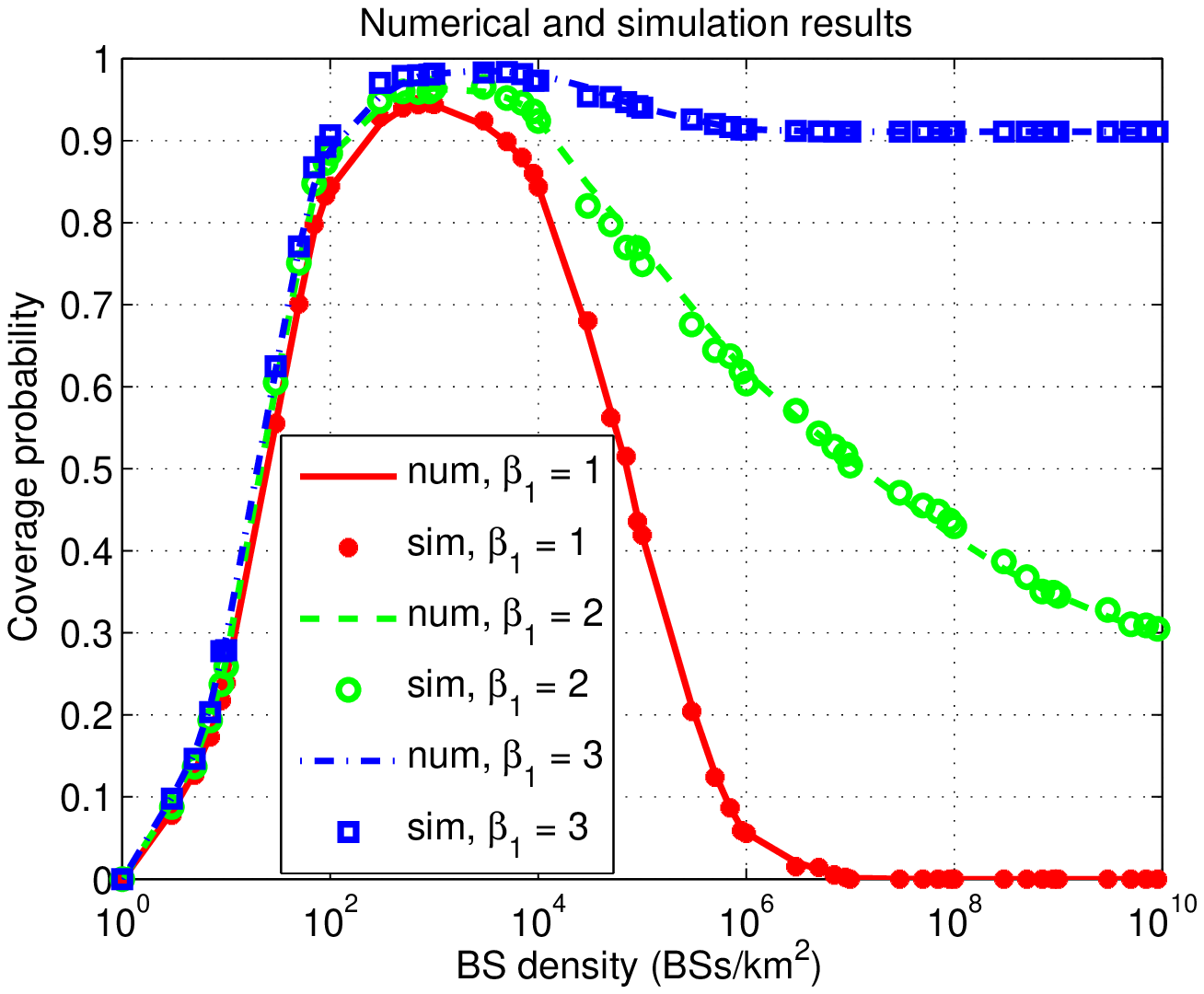}}
    \hspace{-0.1in}
  \subfigure[$A({\lambda},T)$ scaling with $\lambda$]{\label{ASE_general}
    \includegraphics[width=0.235\textwidth]{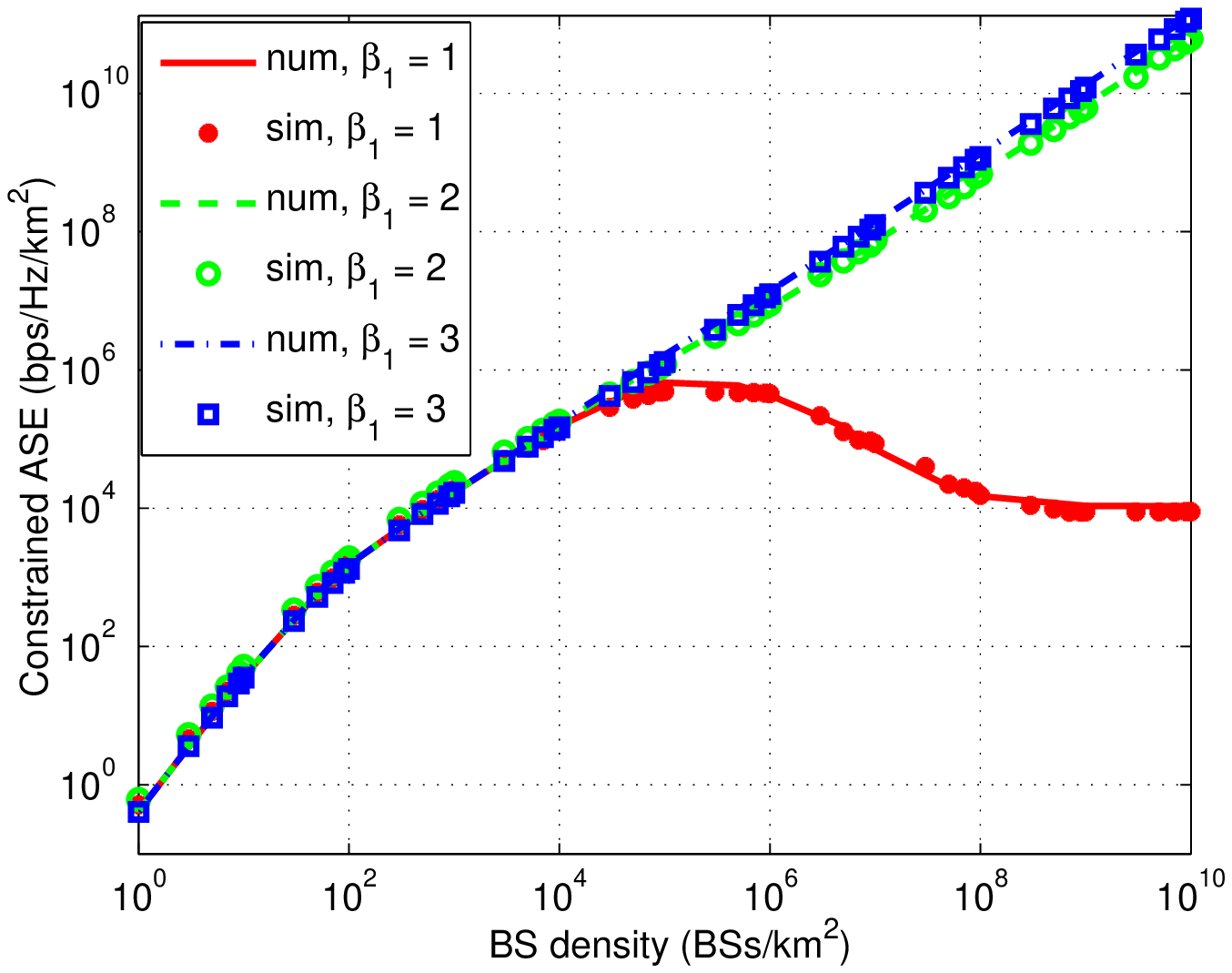}}
    \hspace{0.2in}
  \caption{Coverage probability $P_{c}({\lambda},T)$ and constrained ASE $A({\lambda},T)$ in the general scenario with default system parameters and $\beta_1 = 1,2,3$.}
  \label{general}
\end{figure}

\begin{figure}[!htb]
  \setlength{\abovecaptionskip}{-2pt}
  \setlength{\belowcaptionskip}{0pt}
  \centering
  \vspace{-0.05in}
  \subfigure[$P_{c}({\lambda},T)$ scaling with $\frac{\theta_B \theta_U}{4 \pi^2}$]{\label{Cov_beamalignmentprob}
    \includegraphics[width=0.222\textwidth]{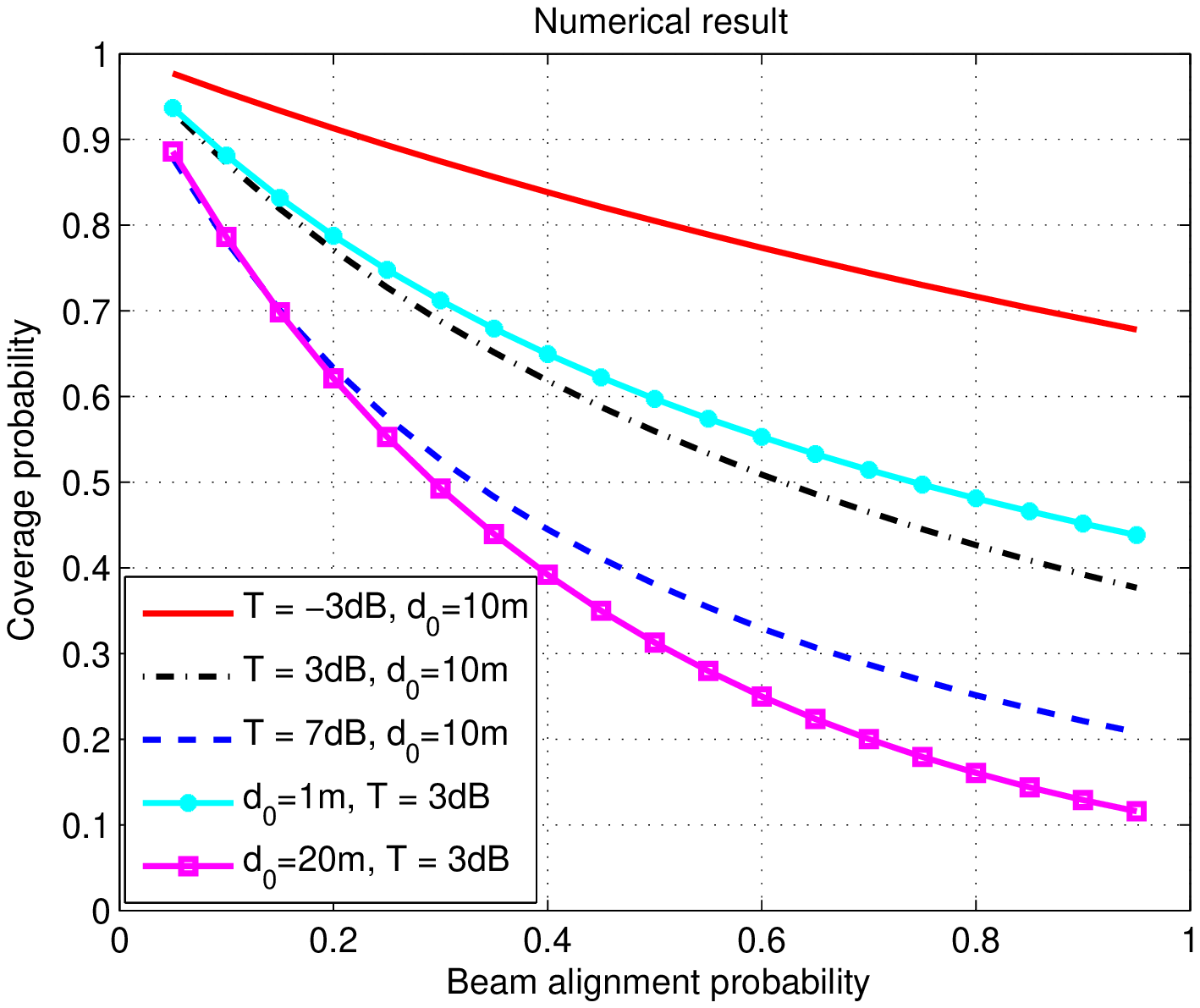}}
    \hspace{-0.1in}
  \subfigure[$A({\lambda},T)$ scaling with $\frac{\theta_B \theta_U}{4 \pi^2}$]{\label{ASE_beamalignmentprob}
    \includegraphics[width=0.225\textwidth]{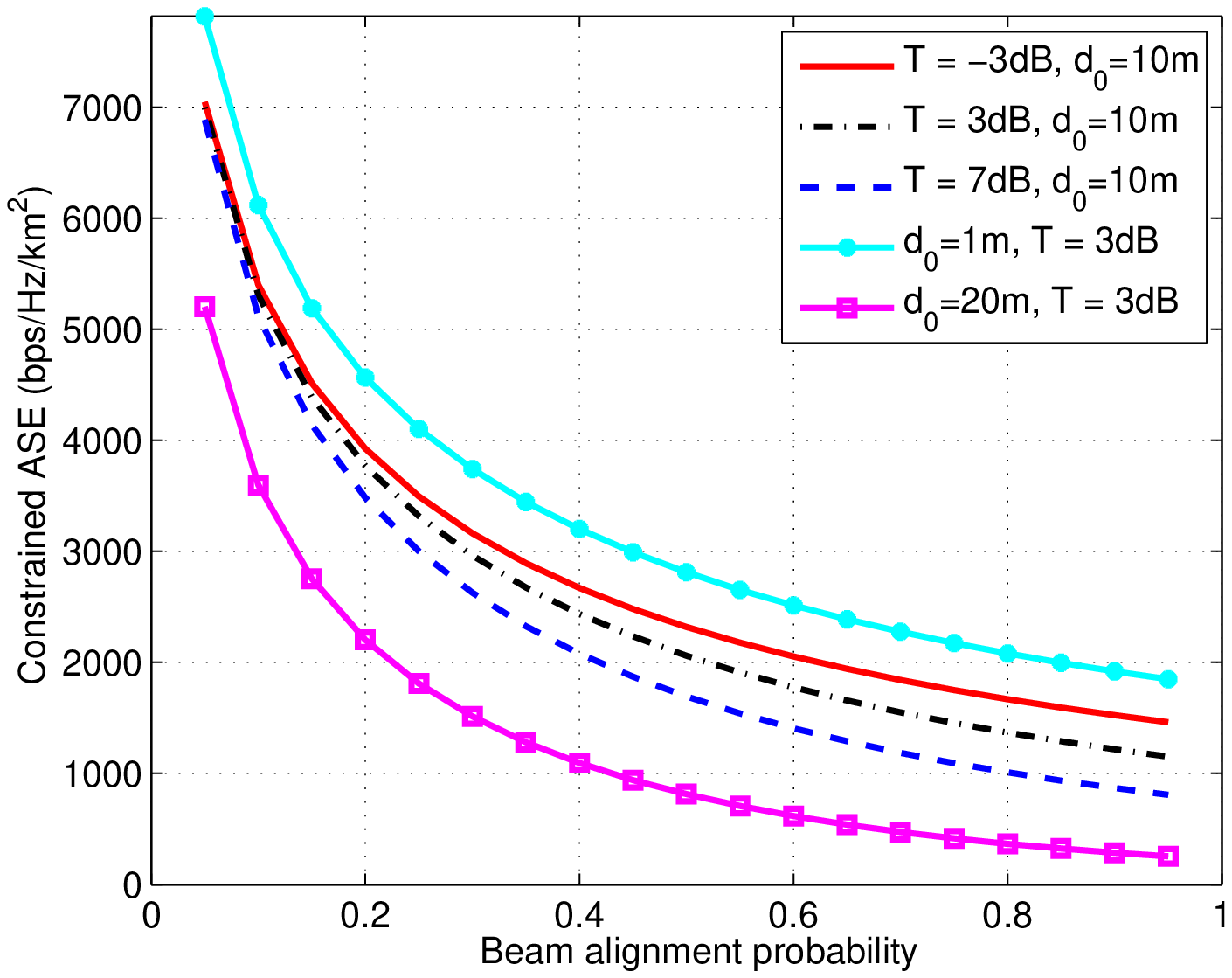}}
    \hspace{0.0in}
  \caption{Coverage probability $P_{c}({\lambda},T)$ and constrained ASE $A({\lambda},T)$ scaling with the beam alignment probability $\theta_B \theta_U/(4 \pi^2)$ in the physically feasible path loss scenario with $\sigma^2 = 0, \beta_1 = 0, \beta_2 = 4, n_B = n_U = 0$ and $\lambda = 1000 ~\text{BSs/km}^2$.}
  \label{beamalignmentprob}
  \vspace{-0.2in}
\end{figure}

\begin{figure}[!ht]
  \setlength{\abovecaptionskip}{-4pt}
  \setlength{\belowcaptionskip}{0pt}
  \centering
  \vspace{0in}
  \subfigure[$P_{c}({\lambda},T)$ scaling with $\lambda$]{\label{Cov_phy_sleep}
    \includegraphics[width=0.227\textwidth]{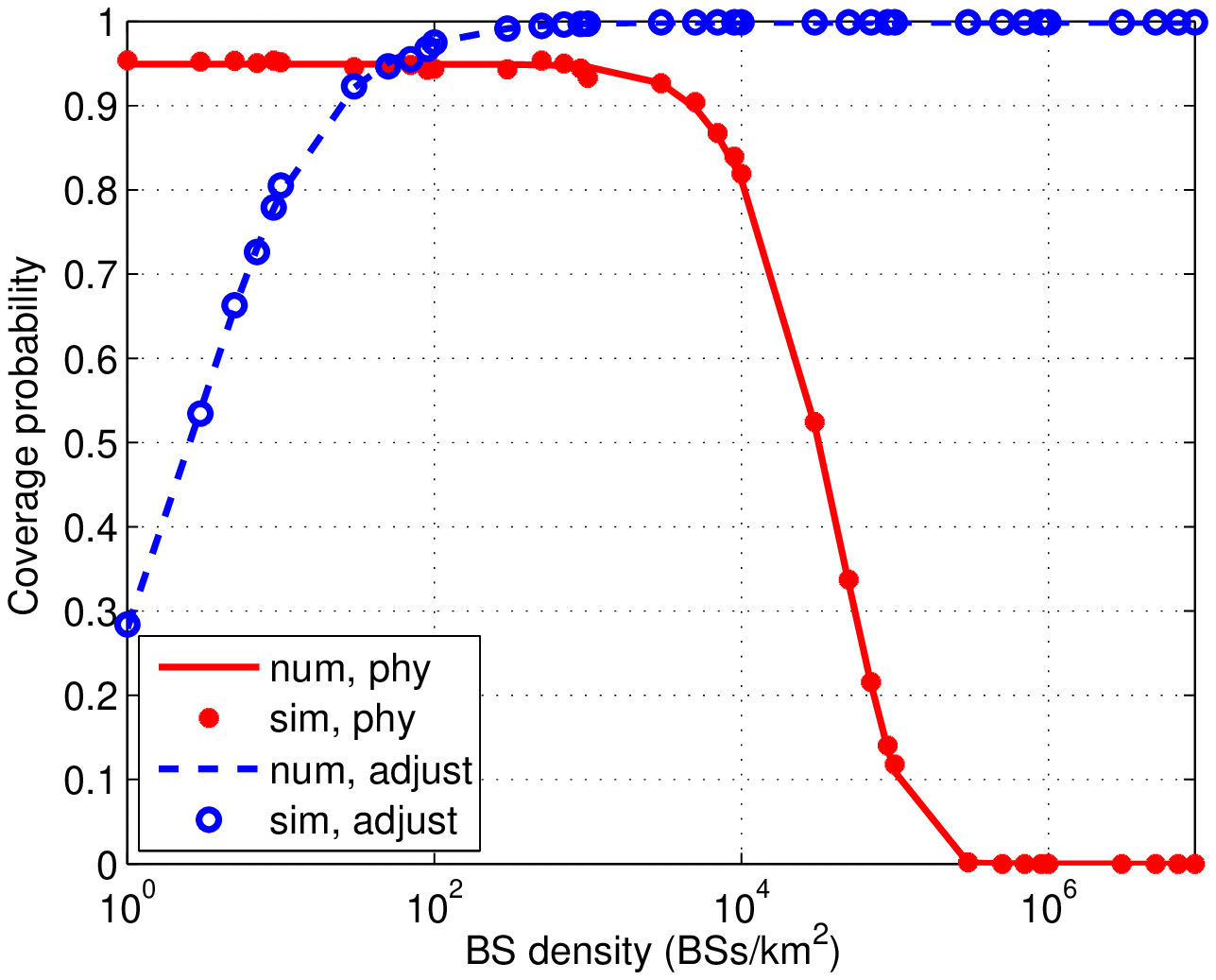}}
    \hspace{-0.1in}
  \subfigure[$A({\lambda},T)$ scaling with $\lambda$]{\label{ASE_phy_sleep}
    \includegraphics[width=0.232\textwidth]{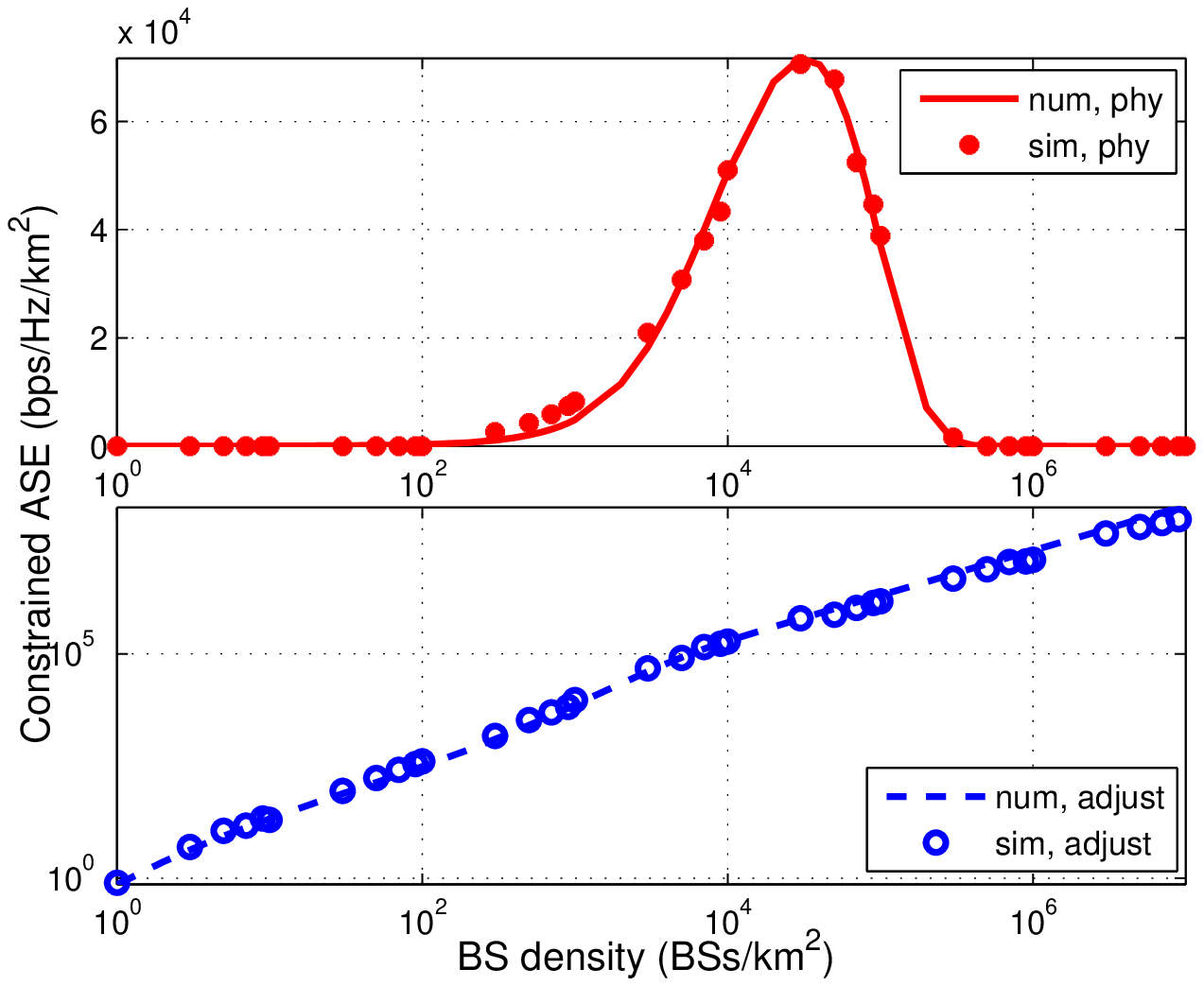}}
    \hspace{0in}
  \caption{Coverage probability $P_{c}({\lambda},T)$ and constrained ASE $A({\lambda},T)$ in the physically feasible path loss scenario and beam pattern adjustment scenario with $\sigma^2 = 0, \beta_1 = 0, \beta_2 = 4$ and $n_B = n_U = 0$.}
  \label{physical_and_sleep}
  \vspace{-0.2in}
\end{figure}

Fig. \ref{Cov_general} shows that when $\beta_1 \leq 2$, the coverage probability has a significant downward trend as proved in Theorem \ref{thm3}. While when $\beta_1 = 3$, the coverage probability approaches a constant when the BS density grows to infinity. In Fig. \ref{ASE_general}: $\beta_1 = 1$ leads to a constant constrained ASE and $\beta_1 = 2,3$ leads to a linearly increasing constrained ASE when the BS density grows to infinity. Note that, even the coverage probability eventually drops to zero, the constrained ASE still keeps increasing when $\beta_1 = 2$. The coverage probability and constrained ASE both increase as $\beta_1$ increases since the interference from BSs within the Fresnel breakpoint dominates the received signal.

An interference-limited network with physically feasible dual-slope path loss model and negligible side lobe gain is discussed as follows. Fig. \ref{beamalignmentprob} shows that the coverage probability and constrained ASE decrease monotonically w.r.t the beam alignment probability under different SINR thresholds $T$ and Fresnel breakpoint distances $d_0$ when $\lambda = 1000~\text{BSs/km}^2$. Here $\frac{\theta_B \theta_U}{4 \pi^2} = 1$ indicates the omni-directional antennas. It shows that the performance is poor when directional transmissions are not used. Conversely, relatively small $\frac{\theta_B \theta_U}{4 \pi^2}$, e.g. $\frac{\theta_B \theta_U}{4 \pi^2} \leq 0.2$, improves the performance notably especially for the constrained ASE. We also find that smaller $d_0$ notably enhances the constrained ASE especially when the beam alignment probability is small. In Fig. \ref{physical_and_sleep}, the coverage probability decreases monotonically with the BS density, while the constrained ASE reaches its maximum and then decreases under the physically feasible path loss model. In the beam pattern adjustment scenario, we set $K=1$. Then according to the equality $\frac{\mathbb{E}[G_i]}{N_B N_U} = \frac{K}{\lambda}$, we adjust the $\frac{\mathbb{E}[G]}{N_B N_U} = \frac{\theta_B \theta_U}{4 \pi^2}$ from 1 to 0, i.e., adjust the beamwidth $\theta_B$ and $\theta_U $ from $2\pi$ to $0$, as the BS density varies from $1 ~\text{BSs/km}^2$ to $9\times 10^6 ~\text{BSs/km}^2$. The result shows that the coverage probability tends to 1 and the constrained ASE grows linearly as the BS density approaches its maximum. The parameter $K$ characterizes the `aggressiveness' of the adaption, while smaller $K$ in general requires more beamforming overhead. In practice, the trade-off between better network performance and lower beamforming overhead should be carefully balanced when selecting $K$.

\section{Conclusion}\label{seccon}
We investigate the coverage and capacity performance of a downlink UDN with directional transmissions using stochastic geometry. The coverage probability is proved to approach zero when the network is extremely dense, the Rayleigh faded interference is considered and the near-field path loss exponent is no larger than 2. The coverage probability and constrained ASE still have a downtrend even under physically feasible path-loss model conditions. However, adapting the beam pattern according to the BS density can provide a constant coverage probability and even linear increase in the constrained ASE when the BS density goes to infinity. Based on this work, further discussions on the impact of beam alignment error and beamforming overhead on the coverage probability and constrained ASE, can be adopted.

\section*{Appendix A: Proof of Theorem 1}
We sketch the major steps due to the space limitation. Applying the distribution of the distance between the typical UE and $B_o$ to the definition of coverage probability, we have the equality (\ref{P_c_1}). For the term inside the first integral, using the property of exponentially distributed $h$, the object is turned to calculate the Laplace transform of the aggregated interference $I$ evaluated at $\mu T/ N_B N_U\alpha_0r_0^{-\beta_1}$ conditioned on $0< r_0<d_0$. Using the i.i.d property of $g_i$ and $G_i$, and applying the probability generating functional of PPP to the object, we can calculate the integral on the left side of the equality (\ref{P_c_2}), where step (\emph{a}) follows that if the nearest BS is within the Fresnel breakpoint, the interfering BSs can be either within or outside the distance $d_0$. Finally, taking similar steps to the second integral in equality (\ref{P_c_1}) leads to the resultant coverage probability.

\section*{Appendix B: Proof of Monotonicity of Coverage Probability w.r.t the Near-field Intensity}
For the coverage probability given in (\ref{simple_cov}), the partial derivative w.r.t the near-field intensity $\lambda \pi d_0^2$ is:
\begin{align*}
 \begin{split}
     \frac{\partial P_c}{\partial (\lambda \pi d_0^2)} =
     & \frac{\frac{\theta_B\theta_U}{4 \pi^2}\!\left(\!\frac{T}{1+T} \!+\! \rho(T,\beta_2)\!\right)}{1-\frac{\theta_B\theta_U}{4 \pi^2}\frac{T}{1+T}}
     \left(\!e^{-\lambda \pi d_0^2 \left(\!\frac{\theta_B\theta_U}{4 \pi^2} \rho(T,\beta_2)+ 1\!\right)}\right.\\
     &\left.- e^{-\lambda \pi d_0^2 \left(\!\frac{\theta_B\theta_U}{4 \pi^2} \rho(T,\beta_2)+\frac{\theta_B\theta_U}{4 \pi^2}\frac{T}{1+T}\!\right)}\!\right),
 \end{split}
\end{align*}
where $\frac{\theta_B\theta_U}{4 \pi^2}\in(0,1]$, $\frac{T}{1+T}\in(0,1)$ and $\rho(T,\beta_2)\in\mathbb{R}_{+}$. Then we have $\frac{\theta_B\theta_U}{4 \pi^2}\frac{T}{1+T} < 1$ and thus the inequality $\frac{\partial P_c}{\partial (\lambda \pi d_0^2)} < 0$ establishes. The monotonicity of $P_c$ w.r.t $\lambda \pi d_0^2$ is proved.

\section*{Appendix C: Proof of Monotonicity of constrained ASE w.r.t the Beam Alignment Probability}
For the constrained ASE given in (\ref{simple_ASE}), it is obvious that the second term inside the integral monotonically decreases w.r.t $\frac{\theta_B\theta_U}{4 \pi^2}$. Denote the partial derivative of the first term w.r.t $\frac{\theta_B\theta_U}{4 \pi^2}$ by $A_1$, and we have:
\begin{align*}
 \begin{split}
 A_1=
 &\frac{\left(\!-\!\lambda \pi d_0^2\!\left(\!\frac{e^t\!-\!1}{e^t} \!+\! \rho(e^t\!-\!1,\beta_2)\!\!\right)\!\!\left(\!1\!-\!\frac{\theta_B\theta_U}{4 \pi^2}\frac{e^t-1}{e^t}\!\right)\!-\!\frac{e^t-1}{e^t}\!\right)\!}
 {(1-\frac{\theta_B\theta_U}{4 \pi^2}\frac{e^t-1}{e^t})^2}\\
 &\quad \cdot e^{-\lambda \pi d_0^2 \frac{\theta_B\theta_U}{4 \pi^2}\left(\!\frac{e^t-1}{e^t}+\rho(e^t\!-1,\beta_2)\!\right)}\\
 &+\frac{\left(\!\lambda \pi d_0^2 \rho(e^t\!-\!1,\beta_2)\!\!\left(\!1\!-\!\frac{\theta_B\theta_U}{4 \pi^2}\frac{e^t-1}{e^t}\!\right)\!+\!\frac{e^t-1}{e^t}\!\right)\!}
 {(1-\frac{\theta_B\theta_U}{4 \pi^2}\frac{e^t-1}{e^t})^2}\\
 &\quad \cdot e^{-\lambda \pi d_0^2 \left(1+\frac{\theta_B\theta_U}{4 \pi^2} \rho(e^t\!-\!1,\beta_2)\right)}\\
 <&\frac{\left(\!-\!\lambda \pi d_0^2\!\left(\!\frac{e^t\!-\!1}{e^t} \!+\! \rho(e^t\!-\!1,\beta_2)\!\!\right)\!\!\left(\!1\!-\!\frac{\theta_B\theta_U}{4 \pi^2}\frac{e^t-1}{e^t}\!\right)\!-\!\frac{e^t-1}{e^t}\!\right)\!}
 {(1-\frac{\theta_B\theta_U}{4 \pi^2}\frac{e^t-1}{e^t})^2}\\
 &\cdot \!\left(\!\!e^{-\lambda \pi d_0^2 \frac{\theta_B\theta_U}{4 \pi^2}\!\left(\!\!\frac{e^t\!-1}{e^t}+\rho(\!e^t\!-\!1,\beta_2\!)\!\!\right)}
 \!\!-e^{-\lambda \pi d_0^2 \left(\!1\!+\frac{\theta_B\theta_U}{4 \pi^2} \rho(\!e^t\!-\!1,\beta_2\!)\!\!\right)}\!\!\right)\\
 <&0,
 \end{split}
\end{align*}
where $\frac{\theta_B\theta_U}{4 \pi^2}\in(0,1]$, $\frac{e^t\!-\!1}{e^t}\in(0,1)$ and $\rho(e^t-1,\beta_2)\in\mathbb{R}_{+}$. Therefore, the first term inside the integral is also monotonically decreasing. And the monotonicity of the constrained ASE w.r.t $\frac{\theta_B\theta_U}{4 \pi^2}$ is established due to the monotonicity of $P_c(\lambda,T)$ and the nature of addition of monotonic functions.

\bibliographystyle{IEEEtran}
\bibliography{bibfile}
\end{document}